%% file: main.tex
\documentclass[sn-mathphys,sn-basic]{sn-jnl}

\usepackage{graphicx}%
\usepackage{multirow}%
\usepackage{amsmath,amssymb,amsfonts}%
\usepackage{amsthm}%
\usepackage{mathrsfs}%
\usepackage[title]{appendix}%
\usepackage{xcolor}%
\usepackage{textcomp}%
\usepackage{manyfoot}%
\usepackage{booktabs}%
\usepackage{algorithm}%
\usepackage{algorithmicx}%
\usepackage{algpseudocode}%
\usepackage{listings}%
\usepackage{hhline}
\usepackage[thinlines]{easytable}

\input{antoine_package.tex}
\newcommand{\CC}{C\nolinebreak\hspace{-.05em}\raisebox{.4ex}{\tiny\bf +}\nolinebreak\hspace{-.10em}\raisebox{.4ex}{\tiny\bf +}}

\theoremstyle{thmstyleone}%
\newtheorem{theorem}{Theorem}
\newtheorem{proposition}[theorem]{Proposition}%

\theoremstyle{thmstyletwo}%
\newtheorem{example}{Example}%
\newtheorem{remark}{Remark}%

\theoremstyle{thmstylethree}%
\newtheorem{definition}{Definition}%

\begin{document}

\title[Temporal Betweenness Centrality on Shortest Walks Variants]{Temporal Betweenness Centrality on Shortest Walks Variants}


\author{\fnm{Mehdi} \sur{Naima}}\email{mehdi.naima@lip6.fr}

\affil{ \orgname{Sorbonne Université, CNRS, LIP6, F-75005 Paris, France}}


\abstract{Betweenness centrality has been extensively studied since its introduction in 1977 as a measure of node importance in graphs. This measure has found use in various applications and has been extended to temporal graphs with time-labeled edges. Recent research by Bu{\ss} et al. \cite{buss2020algorithmic} and Rymar et al. \cite{rymar2021towards} has shown that it is possible to compute the shortest walks betweenness centrality of all nodes in a temporal graph in $O(n^3\,T^2)$ and $O(n^2\,m\,T^2)$ time, respectively, where $T$ is the maximum time, $m$ is the number of temporal edges, and $n$ is the number of nodes. These approaches considered walks that do not take into account contributions from intermediate temporal nodes.

In this paper, we study the temporal betweenness centrality on classical walks that we call \textit{passive}, as well as on a variant that we call \textit{active} walks, which takes into account contributions from all temporal nodes. We present an improved analysis of the running time of the classical algorithm for computing betweenness centrality of all nodes, reducing the time complexity to $O(n\,m\,T+ n^2\,T)$. Furthermore, for active walks, we show that the betweenness centrality can be computed in $O(n\,m\,T+ n^2\,T^2)$. We also show that our results hold for different shortest walks variants.

Finally, we provide an open-source implementation of our algorithms and conduct experiments on several real-world datasets of cities and contact traces. We compare the results of the two variants on both the node and time dimensions of the temporal graph, and we also compare the temporal betweenness centrality to its static counterpart. Our experiments suggest that for the shortest foremost variant looking only at the first $10\%$ of the temporal interaction is a good approximation for the overall top ranked nodes.}




\maketitle

\input{sections/intro}
\input{sections/general_setting}
\input{sections/bet_link}
\input{sections/mining}

\input{sections/appendix}
\input{sections/conclusion}
\bmhead{Acknowledgments}
The author would like to thank Matthieu Latapy and Clémence Magnien for their support and the many discussions we had around this subject. The author is also grateful to Maciej Rymar for taking time to explain their results on this subject and kind conversations.

\newpage

\end{document}

%% file: antoine_package.tex
\usepackage{booktabs}
\usepackage{etex}
\usepackage{multirow}
\usepackage{amsfonts}

\usepackage{pdfpages}
\usepackage{url}
\usepackage{hyperref}
\usepackage{enumerate}
\usepackage{tikz}
\usetikzlibrary{decorations.markings, calc,automata}
\usepackage{stmaryrd}
\usepackage{color}
\usepackage{enumerate}
\usepackage{tikz}
\usepackage{stmaryrd}
\usepackage{color}
\usepackage{pdfpages}

\usepackage[]{algorithm}
\usepackage{etoolbox}\AtBeginEnvironment{algorithmic}{\footnotesize} 
\usepackage{algpseudocode}
\usepackage{mathrsfs}
\usepackage{caption}

\newcommand{\E}{\mathcal{E}}

\newcommand{\V}{\mathcal{V}}

\newcommand{\W}{\mathcal{W}}

\newtheorem{lemma}{Lemma}
\newtheorem{corollary}{Corollary}

\tikzset{
node distance=2cm, 
every state/.style={ fill=gray!10}, 
initial text=$ $ 
}



%% file: sections/intro.tex
\section{Introduction}\label{sec:intro_mon}
Betweenness centrality is a well-known centrality measure in static graphs that aims to identify central nodes in a graph. Centrality measures assign a value to each node (or edge) in based to its importance (centrality). In a static graph, the betweenness centrality of a node is based on the number of shortest paths passing through that node. It was introduced by Freeman in \cite{freeman}. This centrality has been studied extensively in the literature and is a classical measure in network analysis used in a variety of domains such as social networks \cite{burt2004structural}, transports \cite{puzis2013augmented}, biology \cite{narayanan2005betweenness,yoon2006algorithm} and scientific collaboration networks \cite{leyd}. Additionally, betweenness centrality has been utilized as an efficient method for graph partitioning and community detection \cite{girvan2002community}. Brandes in \cite{brandes} introduced a method for computing betweenness centrality of a whole graph in $O(n\,m + n^2)$ which remains the fastest known algorithm. 

Recently, betweenness centrality has been extended to dynamic graph formalisms such as temporal graphs \cite{temporal} and stream graphs \cite{latapy2018stream}. The generalization of betweenness centrality to a temporal setting is not unique, and many optimality criteria have been considered in the literature \cite{buss2020algorithmic,rymar2021towards,sfp,tang,kim2012temporal,latapy2018stream}, including shortest walks, fastest walks, foremost walks, and shortest fastest walks. However, for this paper, we only focus on the shortest walks (minimal number of hops) criteria which has also been studied in \cite{buss2020algorithmic,rymar2021towards} as it is the most straightforward generalization of the static case. It is then possible to define the betweenness centrality of a node $v$ at time $t$ by:
\[ B(v,t) = \sum\limits_{s \neq v \neq z \in V} \dfrac{\sigma_{sz}(v,t)}{\sigma_{sz}}, \]
where $\frac{\sigma_{sz}(v,t)}{\sigma_{sz}}$ is the fraction of shortest temporal walks from $s$ to $z$ that pass through node $v$ at time $t$. Recent results on temporal betweenness centrality, tried with success to adapt Brandes algorithm to the temporal setting \cite{buss2020algorithmic,rymar2021towards}. For shortest walks their approach lead to time complexities of $O(n^3\,T^2)$ and $O(n^2\,m\,T^2)$ to compute the betweenness of a whole temporal graph. However, their algorithms considered only what we call \textit{passive} temporal walks in which the walk only exists when it arrives at a certain temporal node, and moreover, they did not apply Brandes algorithm to its full extent as we shall see. 

In a temporal walk, when there is a delay between steps like starting from a node
$u$, the walk transitions to node $u_1$ at time $t_1$, then later
transitions to node $u_2$ at time $t_2$. Many existing works consider that such a
walk contributes to the betweenness of $u_1$ only at time $t_1$, while we
investigate the more general and more natural version in which the walk
contributes to the betweenness of $u_1$ for all times between $t_1$ and $t_2$.
Indeed, removing $u_1$ at any of these times makes the walk unfeasible. To this end, we consider both what we call \textit{passive} and \textit{active} shortest walks, so that active walks exist all along a node until leaving it while passive walks correspond to the more classical version. For the classical \textbf{passive shortest walks}, we improve the time analysis of \cite{buss2020algorithmic} and show that the Betweenness centrality of the whole graph can be computed in $O(n\,m\,T + n^2\,T)$. This bound increases to $O(n\,m\,T + n^2\,T^2)$ if considering \textbf{active shortest walks}. We also show that these bounds are still true for \textbf{shortest $k$-restless walks} where it is not possible to stay more than $k$ time units on the same node and for \textbf{shortest foremost walks} where we want to reach a node as soon as possible. For all stated criteria the results also hold on their \textit{strict} versions where traversing a node takes one time unit. Our time analysis results show that we can use Brandes approach to its full extent in the temporal setting since when the temporal graph is static (i.e its edges exist at only one timestamp) our analysis reduces to the state of art algorithm on static graphs \cite{brandes}. In fact active walks were considered in \cite{latapy2018stream,tang} but not on shortest walks (number of transitions) and we also seek to have a \textit{systematic} study of all these shortest walks variants as was the case in \cite{buss2020algorithmic,rymar2021towards} by designing a single algorithm for all these variants which was not the aim of these works.

We also provide an open-source implementation in C++ and use it to assess the differences between active and passive variants on real-world temporal graphs in both their \textit{node} and \textit{time} dimensions. On the \textit{node dimension} of the temporal graph we compare temporal betweenness centrality to the static betweenness centrality computed on the aggregated graph. Our experiments show that the temporal and static betweenness centrality rankings of nodes are close to each others with the static betweenness running $100$ times faster. On the \textit{time dimension} our experiments show that the active variant that we propose gives more importance to central times in contrast with the passive classical variant where first the times of the graph are the most important. Finally, our experiments suggest that for the shortest foremost variant looking only at the first $10\%$ of the temporal interactions is a good approximation for the overall top ranked nodes.

The paper is organized as follows, in Section~\ref{sec:form} we introduce our formalism that is a modified version of \cite{rymar2021towards}. We start by defining active and passive walks and giving a motivation for the study of active walks. We end this section by defining the betweenness centrality of a temporal node. In Section~\ref{sec:res} we give the statement of our main Theorem~\ref{thm:gen} followed by a discussion of the results. After that in Section~\ref{sec:proofs} we give the main ideas and algorithms to prove our results with more details given in Section~\ref{sec:appendix}. Finally, Section~\ref{sec:experiments} presents our experimental results. We mainly focus on the differences of behaviours between active and passive walks and show that the rankings of temporal nodes are moderately correlated on real-world datasets. We end this paper with some perspectives in Section~\ref{sec:conclu}.


%% file: sections/general_setting.tex
\section{Formalism}\label{sec:form}
We use a formalism close to the ones used in \cite{buss2020algorithmic,rymar2021towards}. We define a directed temporal graph $G$
as a triple $G = (V, \E, T)$ such that $V$ is the set of vertices, $T \in \mathbb{N}$, is the maximal time step with $[T] := \{1,\dots, T \}$ and $\E \subseteq V \times V \times [T]$ is the set of temporal arcs (transitions). We denote by  $n := |V|$ and $m := |\E|$. We call $V \times [T]$ the set of temporal nodes. Then $(v, w, t) \in \E$ represents a temporal arc from $v$ to $w$ at time $t$. 
\begin{definition}[Temporal walk]
  Given a temporal graph $G = (V, \E, T)$, a temporal walk $W$ is a sequence of transitions $e \in \E^k$ with $k \in \mathbb{N}$, where $e = (e_1,\dots,e_k)$, with $e_i = (u_i,v_i,t_i)$ such that for each $ 1 \leq i \leq k-1$ of $u_{i+1}=v_i$  and $t_i \leq t_{i+1}$.
\end{definition}

The \textbf{length} of a temporal walk $W$ denoted $len(W)$ is its number of transitions. We also denote by $arr(W)$ the time of the last transition of $W$. We can associate a type of walks to consider on a temporal graph. We will study in this paper two types of walks on temporal graphs that are called \textbf{active} and \textbf{passive} walks. We will denote the type of walks considered on a temporal graph $G$ by
$$\mathrm{type}(G) \in \{act,pas \},$$
where $pas$ stands for passive and $act$ for active. The important difference between active and passive walks is that, a passive walk only exists on node transitions, therefore a passive walk that arrive to $v$ at time $t$ and leaves $v$ at $t'$ only exists on node $v$ for a single time $t$, while an active walk exists on $v$ for all times $ t \leq i \leq t'$. This difference is formally defined in the following definition.
\begin{figure}
    \centering
        \includegraphics[scale=0.4]{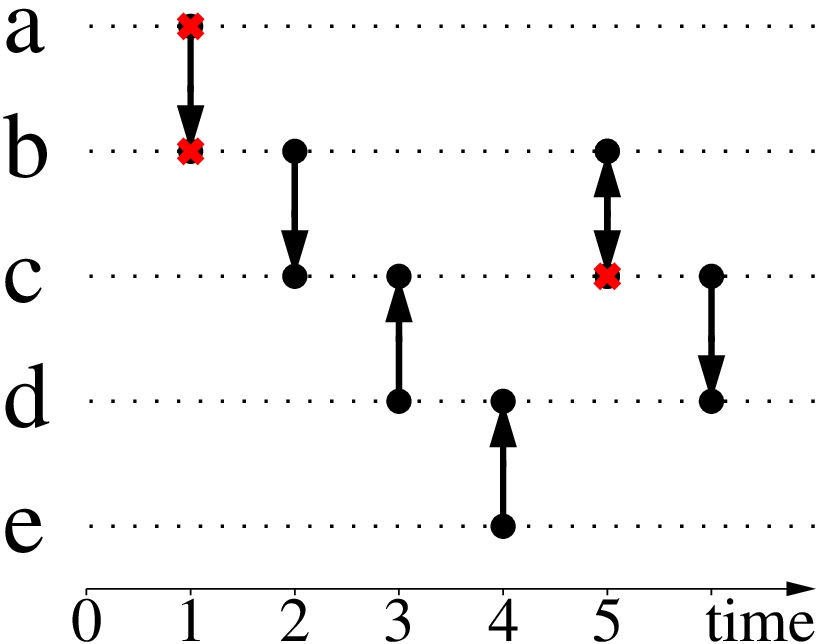}
    \includegraphics[scale=0.4]{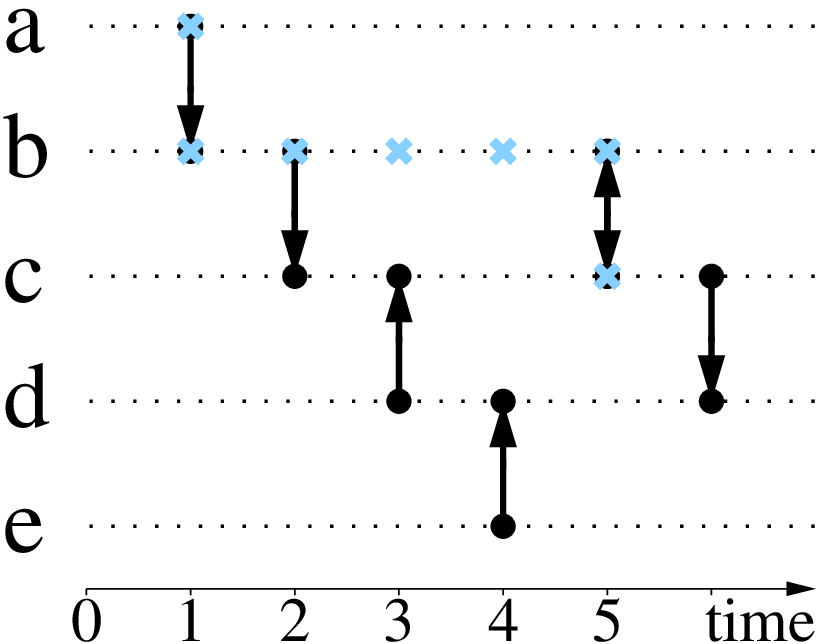}
    \caption{A temporal graph having nodes $V = \{a,b,c,d\}$ and $T = 7$ with arrows representing the set $\E$. The walk $W = [(a,b,1), (b,c,5)]$ can be denoted as $W = a \overset{1}{\rightarrow}  b
\overset{5}{\rightarrow} c$. (left) Walks are passive then $\V(W) = [(a,1),(b,1),(c,5) ]$ marked in red. (right) Walks are active then $\V(W) = [ (a,1),(b,1),(b,2),(b,3),(b,4),(b,5),(c,5)]$ marked in blue.}
    \label{fig:ex}
\end{figure}

\begin{definition}[Visited temporal nodes]\label{def:vis}
 For a temporal graph $G$, fix a walk type. Let $W$ be a temporal walk such that $len(W) = k$ and let $k>0$. Then the list of visited temporal nodes $\V(W)$ is given by:
\begin{equation}
\V(W) =
\begin{cases}
 [ (u_1,t_1)] + [ (v_i,t_i) \,|\,  1 \leq i \leq k ],  & \text{if $\mathrm{type}(G) = pas$ }\\
[(u_1,t_1)] + \left(\biguplus_{i = 1}^{k-1} [(v_i,t) \,|\,t_i \leq t \leq t_{i+1} ]\right)+[ (v_k, t_k)] & \text{otherwise}
\end{cases}    
\end{equation}
where $+$ denotes list concatenation and $\biguplus$ is used for concatenation of several lists.
\end{definition}
We will denote by $W[-1]$ the last visited temporal node corresponding to the last element in $\V(W)$. Let $G = (V,\E,T)$ be a temporal graph with $\mathrm{type}(G) = active$. 
We can denote a temporal walk using an arrow notation. For instance $W = [(a,b,1), (b,c,5)]$ of the temporal graph in Figure~\ref{fig:ex} by $W =  a \overset{1}{\rightarrow}  b \overset{5}{\rightarrow} c$. See Figure~\ref{fig:ex} for an example of these concepts. This distinction between active and passive walks is important since the authors of recent results in this line of research \cite{buss2020algorithmic,rymar2021towards} consider only passive walks.
A temporal walk is called a \textbf{path} if each node $v \in V$ in the list of visited nodes appears exactly once. Moreover, a temporal walk $W$ is a \textbf{strict} temporal walk if for each transition time label is strictly larger than the previous one, that is  for $ 2 \leq i \leq k, t_i > t_{i-1}$. Otherwise, the temporal walk is a \textbf{non-strict} temporal walk. Finally, for $k \in \mathbb{N}$, a temporal walk is \textbf{$k$-restless} if the difference between two consecutive transitions time stamps $t_{i} - t_{i-1} \leq k$.

\begin{example}[Motivation example for the study of active walks]
    Consider a message passing temporal network $G = (V, \E, T)$ where $V$ is a set of machines (computers or routers) and a temporal arc $(u,v,t) \in \E$ corresponds to a message sent from $u$ to $v$ at time $t$. Suppose that Figure~\ref{fig:ex} represents this graph. Take for instance node $c$ at time $4$. At this time node $c$ is retaining a message that arrived from $b$ at time $2$ and another one that arrived from $d$ at time $3$. Therefore, node $c$ at time $4$ is retaining important information. It is important to make sure that machine $c$ is not disconnected from the network at this time to carry this information to other nodes. However, considering the classical passive of temporal walks in \cite{burt2004structural,rymar2021towards} $(c,4)$ is not visited by any temporal walk in the graph and get $0$ value for its betweenness centrality. On the other hand this information is entirely captured by active variant that we propose. 
\end{example}


A walk $W$ is an $s-v$ walk if $W$ starts in node $s$ and ends in node $v$, we denote by $W_{sv}$ the set of all $s-v$ walks. 
In this paper we consider $3$ variants of shortest temporal walks. Shortest here refer to minimizing the number of transitions (length) of the temporal walk. These variants are:
\begin{itemize}
    \item Shortest walks (sh) which minimize the walk length over all walks going from a node to another,
    \item Shortest $k$-restless walks(sh-$k$) which minimize the walk length over all $k$-restless walks going from a node to another and
    \item Shortest foremost (sh-fm) walks which minimize the walk length over all walks going from one node and arriving the earliest in time to the other.
\end{itemize}

For each one of these criteria. We need to define the shortest possible length of temporal walks from a node to another.
\begin{align}
    c_s^{sh}(v) &= \min_{W \in W_{sv}}(len(W)) & \text{(shortest)} \label{eq:cost1} \\
    c_s^{sh-k}(v) &= \min_{ \substack{W \in W_{sv}, \\ W \text{ is } k-restless} }(len(W)) & \text{(shortest $k$-restless)} \label{eq:cost2}\\
    c_s^{fm}(v) &= \min_{W \in W_{sv}}(arr(W)\cdot n + len(W)) & \text{(shortest foremost)} \label{eq:cost3}
\end{align}
Note that the definition of $c_s^{fm}(v)$ ensures only considering walks arriving first and then minimizing over their lengths. 


\begin{remark}
\textbf{Shortest walks} are necessarily paths while this is not true in general for \textbf{shortest $k$-restless walks}. In fact, finding a $k$-restless path has been shown to be NP-hard in \cite{restless}. As a consequence we will use the term walks in general because we want to encompass all variants. 
\end{remark}
For an active temporal walk $W$ we will denote by $W_t$ with $t>arr(W)$ the extension of $W$ on its last node to $t$. Formally, $\mathcal{V}(W_t) = \mathcal{V}(W) + [(v_k,t') \, | \, arr(W) < t' \leq t]$ where $v_k$ is the arrival node of $W$. For example, on the graph of Figure~\ref{fig:ex}, $W = a \overset{1}{\rightarrow}  b \overset{5}{\rightarrow} c$. Then $\mathcal{V}(W_7) = [ (a,1),(b,1),(b,2),(b,3),(b,4),(b,5),(c,5), (c,6), (c,7)]$. Now we can define:
\begin{definition}[Set of shortest walks]  \label{def:opt_walk}
  Let $G = (V, \E, T)$ be a temporal graph and fix a cost. Then
  \begin{equation*}
  \W^{\star} =
  \begin{cases}
  \bigcup_{s,z \in V, s\neq z} \{ W | W \in W_{sz},len(W) = c^{\star}_s(z) ) \}, & \text{ if $type(G) = pas$}    \\
   \bigcup_{s,z \in V, s\neq z} \{ W_T | W \in W_{sz},len(W) = c^{\star}_s(z) ) \}, & \text{ otherwise.}    
  \end{cases} 
  \end{equation*}
  where $\star \in \{ sh,sh-k,fm\}$.
\end{definition}
The reason for the extension of the walks to the last time will be made clear in Section~\ref{sec:appendix}.
\begin{remark}\label{rem:equiv}
    If we allow $k = \infty$ in the $k$-restless setting, then $c_s^{sh}(v) = c_s^{sh-\infty}(v)$ and $\W^{sh} = \W^{sh-\infty}$ since the shortest walk criteria allows all walks regardless of difference in transition time between edges. Therefore, we will only focus on showing our results on $k$-restless criteria for $k \in \mathbb{N} \cup \{ \infty\}$.
\end{remark}

We see that $\W$ is the set of shortest walks between any pair of nodes, it keeps only walks with an overall shortest value. 

\begin{definition}
  Let $G = (V, \E, T )$ be a temporal graph. Fix a walk type, a cost and let $s, v, z \in V$ and $t  \in [T ]$. Let $\W^{\star}$ be as in Definition~\ref{def:opt_walk}. Then,
  \begin{itemize}
  \item $\sigma_{sz}$ is the number of $s-z$ walks in $\W^{\star}$.
  \item $\sigma_{sz}(v,t)$ is the number of $s-z$ walks $W \in \W^{\star}$, such that $W$ passes through $(v,t)$ that is $(v,t) \in \V(W)$ according to Definition~\ref{def:vis}.
  \end{itemize}
\end{definition}
We note that $\sigma_{sz}$ depends only on the cost considered while $\sigma_{sz}(v,t)$ depends on both the cost and the walk type considered.
\begin{definition}
Given a temporal graph $G = (V, \E, T )$, a walk type and a cost. We define
  \begin{align*}
    {\delta}_{sz}(v,t) = \begin{cases}
      0 &\text{ if $\sigma_{s z}=0$ ,}\\
      \dfrac{\sigma_{sz}(v,t)}{\sigma_{sz}} &\text{ otherwise.}
    \end{cases} & &
                    \delta_{s\bullet}(v,t) = \sum_{z \in V}\delta_{sz}(v,t).
  \end{align*}
\end{definition}

\begin{definition}[Betweenness centrality of a temporal node]\label{bet:1}
Given a temporal graph $G = (V, \E, T )$ and a walk type. The betweenness centrality of node $v$ at time $t$ is:
  \begin{equation}\label{eq:bet}
      B(v,t) = \sum_{\substack{s,v,z \in V,\\ s\neq v \neq z } } {\delta}_{sz}(v,t).
  \end{equation}
\end{definition}
According to our definitions there are $2$ walk types and $3$ costs considered. Therefore, we have $6$ different variants that can be considered corresponding to any combination of walk type and cost. Now, from the preceding we define 
\[\hat{B}(v,t) = \sum_{s,z \in V} {\delta}_{sz}(v,t) \implies \quad \hat{B}(v,t) = \sum_{s \in V} {\delta}_{s \bullet}(v,t).\] 
The quantities $B(v,t)$ and $\hat{B}(v,t)$ are related through:
\begin{equation}\label{eq:hat_b}
    \begin{split}
    B(v,t) &= \hat{B}(v,t) - \sum_{w \in V}(\delta_{vw}(v,t) +
  \delta_{wv}(v,t)) = \hat{B}(v,t) - \delta_{v \bullet}(v,t) -  \sum_{w \in V}\delta_{wv}(v,t)  
    \end{split}
\end{equation}
For instance on Figure~\ref{fig:ex}, considering \textbf{passive} shortest walks (approach used in \cite{buss2020algorithmic,rymar2021towards}), $B(b,1) = 2$, and  $\forall t > 1, B(b,t) = 0$ while if we consider \textbf{active} shortest walks we have $B(b,1) = B(b,2)= 2, B(b,3) = B(b,4) = 1$ showing that the \textbf{active} version takes into account contributions from intermediate temporal nodes in shortest paths while it is not true for the \textbf{passive} version.

From the betwenness centrality of a temporal node we can get an overall betweenness centrality of a node $v$ and an overall betweenness centrality of a time.
\begin{definition}[Overall betweenness of a node and overall betweenness of a time]\label{def:time_node}
\begin{equation}\label{eq:overall}
 B(v) = \sum_{t \in [T]} B(v,t), \quad B(t) = \sum_{v \in V} B(v,t).  
\end{equation}    
\end{definition}
\section{Results}\label{sec:res}
While the authors of \cite{buss2020algorithmic,rymar2021towards} focus on computing $B(v)$ for all $v \in V$, we focus on the computation of $B(v,t)$.
Our main result is the following:
\begin{theorem}\label{thm:gen}
  Let $G = (V,\E,T)$ be a temporal graph. For passive walks, the betweenness centrality of all temporal nodes can be computed in $O(n^2\,T\,+ n\,m\,T)$ considering shortest, shortest $k$-restless and shortest foremost walks. For active walks, the betweenness centrality of all temporal nodes can be computed in $O(n^2\,T^2\,+ n\,m\,T)$ considering shortest and shortest $k$-restless walks. Both results hold for strict and non-strict versions.
\end{theorem}
\begin{proof}
We leave the proof of Theorem~\ref{thm:gen} to the end of Section~\ref{sec:proofs}.
\end{proof}
\begin{table}[h]
\centering
    \begin{tabular}{lccc}
    \hline
            & \cite{rymar2021towards}  & \cite{buss2020algorithmic} & Theorem~\ref{thm:gen}\\
    \hline
        Shortest  (passive)&  $O(n^2\,m\, T^2)$ & $O(n^3\, T^2)$ & $O(n^2T\,+ nmT )$\\
        Shortest $k$-restless (passive)&  $O(n^2\,m\, T^2)$ & - & $O(n^2T\,+ nmT )$\\
        Shortest foremost (passive) &  $O(n^2\,m\, T^2)$ & $O(n^3\, T^2)$ & $O(n^2T\,+ nmT )$\\
        Shortest  (active)& -  & - & $O(n^2T^2\,+ nmT )$\\
        Shortest $k$-restless  (active)& -  & - & $O(n^2T^2\,+ nmT )$\\
    \end{tabular}
    \caption{Improvement of previously known results by Theorem~\ref{thm:gen}. The results on active walks were not studied in this form to our knowledge. All results hold for non-strict and strict variants.}
    \label{tab:res}
\end{table}

\textbf{Discussion}. The authors of \cite{buss2020algorithmic,rymar2021towards} showed that for passive walks, the overall betweenness of nodes $B(v)$ (not temporal nodes) can be computed in $O(n^3\,T^2)$ and $O(n^2\,m\, T^2)$ respectively. Since the maximal number of temporal arcs is $(n-1)^2T$, our bounds are always better than the previously known ones. Additionally, in the introduction we mentioned using \textit{Brandes approach to its full extent}, since these previous approaches when $T=1$ reduce to $O(n^3)$ and $O(n^2\,m)$ while our analysis lead to $O(nm + n^2)$. Therefore our approach leads to the static optimal time algorithm if the temporal graph is static. Table~\ref{tab:res} summarises our results compared to the other two when taking the overall betweenness of nodes. 

It is worth to note that out of the $6$ variants mentioned only the active version of shortest foremost can not be computed using our algorithm. In Section~\ref{sec:conclu} we discuss why the same result on active shortest foremost walks does not hold. 

%% file: sections/bet_link.tex
\section{Main algorithms and proofs}\label{sec:proofs}
According to Remark~\ref{rem:equiv} we only need to consider $3$ variants in our proofs that are (active, $k$-restless), (passive, $k$-restless) and (passive, shortest foremost) since  $k=\infty$ covers the classical shortest walks criteria. We will prove our results for both (active, $k$-restless), (passive, $k$-restless) walks and in Section~\ref{sec:shfm} give the necessary details for (passive, shortest foremost) variant.

We denote by $W^{pas,k}_{s(v,t)}$ the set of \textbf{passive} $k$-restless $s-(v,t)$ walks and by $W^{act,k}_{s(v,t)}$ the set of \textbf{active} $k$-restless $s-(v,t)$ walks. These are defined as: 
\begin{align*}
  W^{k,pas}_{s(v,t)} =&  \{W | W \in W_{sv}, arr(W) = t, W \text{ is $k$-restless }  \}\\ & \cup \{ \epsilon \,|\, s=v, (s,w,t) \in \E \text{ for some $w \in V$} \},  
\end{align*}
where $\epsilon$ denotes the empty walk. For active walks there can be two types of walks, either the last transition of the walk is $(w,v,t)$ which we call an exact-$s-(v,t)$ walk or, the walk arrived earlier to $v$ at time $t' < t$. Then:
\begin{align*}
    W^{k,act}_{s(v,t)} =&  \{W | W \in W_{sv}, arr(W) \leq t, W \text{ is $k$-restless }\}\\ & \cup \{ \epsilon \,|\, s=v, (s,w,t) \in \E \text{ for some $w \in V$}\}.
\end{align*}

Let $G = (V, \E, T )$ be a temporal graph. Fix a source $s \in V$ and a walk type. Then, for every temporal node $(v, t) \in V \times [T]$  we define the optimal cost from $s$ to temporal node $(v,t)$. Formally:

\[c^{k,*}_s(v,t) = \min\limits_{ W \in W_{s(v,t)}^{*,k} } (len(W)),\]

where $* \in \{ pas, act\}$. If the set of $s-v$ walks is empty then $c^{k,*}_s(v,t) = c^{fm,pas}_s(v,t) = \infty$. Now, the overall optimal values from $s$ to any time on node $v$ as defined in Equation~\eqref{eq:cost2} can be computed as:
\[c_s^{k,*}(v) = \min_{t \in [T]} (c^{k,*}_s(v,t)).\]

Finally, for fixed walk type we say that a temporal $s-(v,t)$ walk $W$ is an \textbf{optimal} $s-(v,t)$ $k$-restless walk if $c(W) = c^{k,*}_s(v,t)$. Similarly an $s-z$ walk $W$ is an \textbf{optimal} $s-z$ $k$-restless walk if  $c(W) = c^{k,*}_s(z)$. In the active type, a walk $W$ can be optimal to different times on $v$. For instance on the graph of Figure~\ref{fig:ex}. The walk $W = a \overset{1}{\rightarrow}  b
\overset{2}{\rightarrow} c$, ${W}$ is an optimal $a-(c,2)$ walk and ${W}$ is also an optimal $a-(c,4)$ walk.

\begin{table}[]
    \centering
\begin{tabular}{c|ccccccc}
\hline
 $t$ & $0$ & $1$ & $2$ & $3$ & $4$ & $5$ & $6$ \\ \hline
 $c^{\infty,pas}_a(b,t)$ & $\infty$ & $1$ & $\infty$ & $\infty$ & $\infty$ & $3$ & $\infty$ \\
 $c^{\infty,act}_a(b,t)$ & $\infty$ & $1$ & $1$ & $1$ & $1$ & $1$ & $1$ \\
 & & & & & & & \\
 $v$ & $a$ & $b$ & $c$ & $d$ & $e$ & & \\ \hline
$c^{sh}_a(v)$ & $0$ & $1$ & $2$ & $3$ & $\infty$ &  &  \\
 $c^{sh}_a(v)$ & $0$ & $1$ & $2$ & $3$ & $\infty$ &  &   
\end{tabular}
    \caption{Values of $c_s^{\infty,*}(v,t)$ and $c^{sh}_s(v)$ on the temporal graph of Figure~\ref{fig:ex}. (Upper part) values of $c^{\infty,*}_a(b,t)$ for passive walks (1st row), and for active walks (2nd row) and for all $t \in [T]$. (Lower part) Overall optimal values of $c^{sh}_a(v)$ for all $v \in V$.}
    \label{tab:c}
\end{table}

The two major steps of the proof are the following. First step is to build a predecessor graph from a fixed node $s\in V$ efficiently. This predecessor graph allows then to compute the contributions of node $s \in V$ to the betweenness centrality of all other nodes. Second step is to find a recurrence that allows to compute the aforementioned contributions efficiently. Proposition~\ref{prop:com_pred} and Proposition~\ref{prop:rec_contri} correspond to these steps. For all quantities defined in the paper when the distinction between active and passive walks is needed we write pas for passive and act for active in the superscript like $\delta^{act}_{ab}(v,t)$, means we want $\delta_{ab}(v,t)$ for active walks on the considered temporal graph. Otherwise we will drop the superscript when not necessary for instance writing $c_s(v,t)$ instead of $c^{k,*}_s(v,t)$.


Given a temporal graph and walk type. We denote by $\W^{k,*}_{s(v,t)}$ the set of \textbf{optimal $s-(v,t)$ walks}. That is 
\[ \W^{k,*}_{s(v,t)} = \{W | W \in W^{k,*}_{s(v,t)}, len(W) = c^{k,*}_{s}(v,t).\}\]

\begin{definition}[Exact optimal $s-(v,t)$ walks]\label{def:exact_svt}
We say that an $s-(v,t)$ walk $W$ is an \textit{exact optimal $s-(v,t)$ walk} if $W \in \W^{k,*}_{s(v,t)}$ and $arr(W) = t$. We denote by $\overline{\W}^{k,*}_{s(v,t)}$ the set of all \textit{exact optimal $s-(v,t)$ walks}. 
\end{definition}
Note that for passive walks $\overline{\W}^{k,*}_{s,(v,t)}$ and ${\W}^{k,*}_{s,(v,t)}$ coincide while this is not true in general for active walks. A consequence of our framework is that the empty walk is an exact $s-(s,t)$ whenever there exists at least one node $w \in V$ with $(s,w,t) \in \E$. Then $c_s(s,t) = 0$. The predecessor graph in temporal settings has been used in  \cite{rymar2021towards,buss2020algorithmic}. Here, we extend its definition to encompass active walks as well.

In order to simplify notations we will drop the superscript that is used to identify $k$ and the walk type. For instance we will write $\overline{\W}_{s(v,t)}$ instead of $\overline{\W}^{k,*}_{s(v,t)}$ and ${\W}_{s(v,t)}$ instead of ${\W}^{k,*}_{s(v,t)}$. The superscript will be added whenever necessary to distinguish between the walk types active and passive. However, it should be clear to the reader that $k$ and the walk type are always fixed.

\begin{definition}[predecessor set, successor set]\label{def:pred}
  Let $G = (V,\E,T)$ be a temporal graph,  fix a walk type and a source $s \in V$, then for all $w \in V, w \neq s$, let $\overline{\W}_{s,(v,t)}$ be the set of exact optimal $s-(v,t)$ walks:
  \begin{align*}
    pre_s(w,t') = &\{ (v,t) \in V \times [T] \mid \exists\,  m \in \overline{\W}_{s(w,t')},  m = s \overset{t_1}{\rightarrow}  \dots \overset{t}{\rightarrow}  v   \overset{t'}{\rightarrow} w \}\\ & \cup \{(s,t') \mid \exists\,  m \in \overline{\W}_{s(w,t')},  m = s \overset{t'}{\rightarrow} w \}.
  \end{align*}
The successor set of a node $succ_s(w,t') = \{(v,t) \,|\, (w,t') \in pre_s(v,t) \}$.
\end{definition}

\begin{definition}[Predecessor graph]
  The predecessor graph $G_s = (V_s,E_s)$ is the directed graph obtained from $pre_s$, whose arcs are given by
  \[ E_s = \{((v,t),(w,t')) \mid (v,t) \in pre_s(w,t') \},\]
  and its vertices $V_s$ are the ones induced by $E_s$.
\end{definition}

\begin{figure}
    \centering
    \includegraphics[scale=0.27]{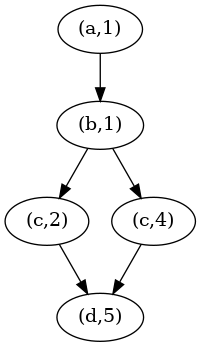} \includegraphics[scale=0.27]{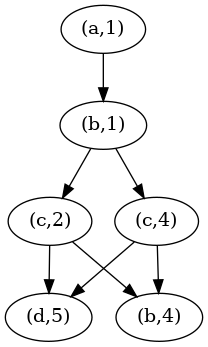}
    \caption{The predecessor graphs of shortest paths from node $a$ on the temporal graph of Figure~\ref{fig:ex}. (left) the walks are considered active and (right) the walks are considered passive.}
    \label{fig:pred}
\end{figure}
An example of the predecessor graph for active and passive walks on shortest walks (i.e $k=\infty$) is depicted on Figure~\ref{fig:pred}. As we shall see, a path in the predecessor graph represents a unique walk in the temporal graph. The next Proposition gives the relationship between these quantities.
\begin{remark}\label{rem:small}
    \item The predecessor graph $G$ from node $s$ in the active case is a subgraph of the predecessor graph $G'$ in the passive case. See Figure~\ref{fig:pred} for an example. 
\end{remark}

\begin{lemma}\label{lem:bij}
    Fix $s \in V$.  There is a one-to-one correspondence between a path $p$ in the predecessor graph $G_s$ starting from node $(s,t)$ for some $t \in [T]$ and ending in $(v,t')$ and an exact shortest $s-(v,t')$ walk.
\end{lemma}

\begin{lemma}\label{lem:dag}
    The predecessor graph $G_s$ from $s \in V$ is acyclic.
\end{lemma}
We define $\overline{\sigma}_{s,(v,t)} = | \overline{\W}_{s,(v,t)}|$ which corresponds to the number of optimal exact $s-(v,t)$ walks. We also denote by $\sigma_{s(v,t)}$ the number of optimal $s-(v,t)$ walks, then $\sigma_{s(v,t)} = |\W_{s(v,t)}|$.
\begin{proposition}\label{prop:overline_sig} For any temporal node $(v,t)$, there holds that:
  \begin{equation}\label{eq:sigma_overline}
    \overline{\sigma}_{s(v,t)} =  \begin{cases}
    0 &\text{ if $(v,t) \neq V_s$},\\
      1 &\text{ if $(v,t) \in V_s$ and $v = s$},\\
      \sum\limits_{(w,t') \in pre_s(v,t)} \overline{\sigma}_{s(w,t')} &\text{ otherwise.}
    \end{cases}
  \end{equation}
\end{proposition}

In the next proposition we need to distinguish between the two walk types. 
\begin{proposition}\label{prop:sig}
Let $s \in V$, $G_s = (V_s,E_s)$ be the predecessor graph from node $s$ and for any temporal node $(v,t) \in V \times [T]$:
\begin{equation}\label{eq:sigma}
\sigma^{pas}_{s(v,t)} = \begin{cases}
    \overline{\sigma}^{pas}_{s(v,t)} & \text{if $(v,t) \in V_s$},\\
    0 & \text{otherwise}.
\end{cases}, \quad \sigma^{act}_{s(v,t)} = \sum_{\substack{t' \leq t, \exists W \in \overline{\W}_{s,(v,t')} \\
        c(\bar{W}_t) = c_s(v,t) }} \overline{\sigma}^{act}_{s(v,t')}
\end{equation}

\end{proposition}
Finally it is straight forward to compute for all $v\in V$ and $t \in [T]$, $\sigma_{sv}$ and $\sigma_{sv}(v,t)$ from the preceding. The same relation applies for both types:
\begin{equation}\label{eq:sigma_sz}
    \sigma_{sv} = \sum_{\substack{t \in [T] \\ c_s(v,t) = c_s(v)} }\overline{\sigma}_{s(v,t)}, \quad {\sigma}_{sv}(v,t) = \begin{cases}
     {\sigma}_{s(v,t)} & \text{if $c_s(v,t) = c_s(v)$}\\
     0 & \text{otherwise}.
\end{cases}
\end{equation}


We use a temporal BFS algorithm variant of the one used in \cite{buss2020algorithmic}. The relaxing technique builds the shortest $s-(v,t)$ walks and for active walks it checks if the extension of a walk arriving to $(w,t')$ is also shortest to $(w,t'')$ with $t'' > t'$. The procedure is defined in function \textsc{relax} of Algorithm~\ref{algo:BFS}.

\begin{algorithm}[h]
	\caption{\label{algo:BFS} Predecessor graph from node $s$}
	\begin{algorithmic}[1]
    \renewcommand{\algorithmicrequire}{\textbf{Input:}}
    \renewcommand{\algorithmicensure}{\textbf{Output:}}
    \Require  $G = (V,\E,T)$ : a temporal graph, $s$ : a node in $G$, $k$ : maximal waiting time for $k$-restless walks. 
    \Ensure A dictionary $dist$ containing shortest values $c_s(v,t)$ to
    temporal nodes. A dictionary $pre$ that
    contains the set of predecessor temporal nodes.
    \Function{Temporal\_BFS}{$G$,$s$,$k$}
    \State $pre, dist, Q= \textsc{initalization}(G,s)$
    \State $Q' = \textsc{Empty\_Queue()},\ell = 1$
    \While {$Q \neq \emptyset$} \label{line:loop}
    \For{$(a,t)$ in $Q$}

    \For {$a \overset{t'}{\rightarrow} b \in \E$ such that (not ($(s=a)$ and $(t' \neq t)$))}\label{line:arc} 
        \If{$(s = a)\, or \, ((t' \geq t) \,and\, (t'-t)\leq k)$} \label{line:strict}
\State $\textsc{relax}(a,b,t,t',pre,dist, Q',\ell, k)$
        \EndIf
    
    \EndFor
    \EndFor
    \State $\ell = \ell +1$
    \State $(Q,Q') = (Q',\textsc{Empty\_Queue()})$
    \EndWhile
    \State \textbf{return} $pre, dist$
    \EndFunction
  \end{algorithmic}
  	\begin{algorithmic}[1]
    \renewcommand{\algorithmicrequire}{\textbf{Input:}}
    \renewcommand{\algorithmicensure}{\textbf{Output:}}

    \Function{initialization}{$G$,$s$}
    \State $Q = \textsc{Empty\_Queue}()$
    \State $\mathrm{dist}[v] = \{t:\infty,   \forall t \in  \{0,\dots,T\} \}$ $\forall v \in V$

    \State $\mathrm{pre}[v] = \{t:\{\},   \forall t \in \{0,\dots, T \} \}$ $\forall v \in V$
     \For { {$t \in \{t'\,|\,\exists w \in V,s \overset{t'}{\rightarrow} w \in \E \}$}  }
     \State {$dist[s][t] = 0$, $pre[s][t] = \{(nil,nil)\}$}
     \State $Q.\textsc{enqueue}( (s,t) )$
     \EndFor    
    \State \textbf{return} $pre, dist, Q$    
    \EndFunction
      \end{algorithmic}
      	\begin{algorithmic}[1]
       \Function{relax}{$a,b,t,t',dist,pre, Q, \ell$, $k$}
    \If {($dist[b][t'] = \infty$) or ($dist[b][t'] \geq \ell$ and $|\,pre[b][tp]\,| = 0$)} \label{line:cond}
    \State $dist[b][t'] = \ell$, $pre[b][t'] = \{ \}$ \label{relax:cur1}
    \State $Q.\textsc{enqueue}((b,t'))$
    \For {$t'' \in \{r\,|\,\exists w,w \overset{r}{\rightarrow} b \in \E, (r>t'\, and \, (r-t') \leq k) \}$} \label{line:arc2} \Comment{for passive walks, ignore this loop.}
    \State $\textsc{relax\_extend}(b,t'',pre,dist,\ell)$ 

    \EndFor
    
    \EndIf
    \If {$dist[b][t'] = \ell$} \label{line:pred}
    \State $pre[b][t'].\textsc{add}((a,t))$
    \EndIf
    \EndFunction
    \Function{Relax\_extend}{$b,t',dist,pre,\ell$}
    \If {$dist[b][t'] > \ell $} 
    \State $dist[b][t'] = \ell$, $pre[b][t'] = \{ \}$ 
    \EndIf
    \EndFunction
  \end{algorithmic}
\end{algorithm}
    \begin{definition}[Exactly reachable temporal nodes] Let $G = (V,\E,T)$ be a temporal graph, fix a walk type and a source $s \in V$. Then we define:
    \[ER_s = \{ (v,t) \, | \, (v,t) \in V \times [T], \overline{\W}_{s(v,t)} \neq \emptyset\}.\]
\end{definition}

\begin{proposition}\label{prop:BFS}
  Algorithm \textsc{Temporal\_BFS} solves the shortest walk problem
  for a temporal graph $G=(V,\E,T)$. That is For
  all $(v,t) \in ER_s$, $dist[v][t] = c_s(v,t)$, $pre[v][t] = pre_s(v,t)$ and $(v,t)$ is added exactly once in the queue $Q$.
\end{proposition}
\begin{proposition}\label{prop:com_pred}
  Let $G = (V,\E,T)$ be a temporal graph, fix a walk type and a source $s\in V$, then the predecessor graph $G_s$ can be computed in $O( m\,T+n\,T)$.
\end{proposition}

\begin{corollary}\label{cor:quan}
For both walk types and for all $v \in V$ and $t \in [T]$ the quantities $\sigma_{s(v,t)}$ and $\delta_{sv}(v,t)$ can be computed for all temporal nodes $(v,t)$ in $O(n\,T + m\,T)$. 
\end{corollary}

The main result allowing to efficiently compute the contributions from a node $s$ to the betweenness centrality of all others is a an extension of the recurrence found by Brandes in \cite{brandes}. This recurrence has been adapted to temporal graphs in \cite{buss2020algorithmic,rymar2021towards} and here we extend it further for active walks. We first define $before_{G_s}(v,t)$ to be the largest time $t'$ such that $t' \leq t$ and $(v,t') \in V_s$. Therefore if $(v,t) \in V_s$, then $before_{G_s}(v,t) = t$. Let $G_s =(V_s,E_s)$ be the predecessor graph of node $s$, then:
\begin{proposition}[General contribution]\label{prop:rec_contri}
Fix a node $s \in V$, a walk type, then if we are considering active walks for any temporal node $(v,t) \in V \times [T]$:
  \begin{equation}\label{eq:rec_act}
      \delta^{act}_{s \bullet}(v,t) = \delta^{act}_{sv}(v,t) + \sum\limits_{\substack{ t'':= before_{G_s}(v,t)\\ (w,t') \in succ_s
      (v,t'') \\ t' \geq t}}\dfrac{\sigma^{act}_{s(v,t)}}{\sigma^{act}_{s(w,t')}} \delta^{act}_{s
    \bullet}(w,t'),
  \end{equation}
  
  if we are considering passive walks then for $(v,t) \in V_s$
\begin{equation}\label{eq:rec_pas}
  \delta^{pas}_{s \bullet}(v,t) = \delta^{pas}_{sv}(v,t) + \sum\limits_{\substack{(w,t') \in succ_{s}
      (v,t) \\ t' \geq t}}\dfrac{\sigma^{pas}_{s(v,t)}}{\sigma^{pas}_{s(w,t')}} \delta^{pas}_{s \bullet}(w,t').  
\end{equation}
\end{proposition}

Proposition~\ref{prop:rec_contri} allows to compute the values of $\delta_{s\bullet}$ by recurrence for all temporal nodes by starting the recurrence from the sources of $G_s$. Finally, to compute the betweenness centrality of the whole temporal graph it suffices to sum $\delta_{s\bullet}$ for all $s\in V$ and use the correction formula given in Equation~\eqref{eq:hat_b} needed to go from $\hat{B}(v,t)$ to $B(v,t)$. These steps are summarised in Algorithm~\ref{algo:betweenness}. In the Algorithm function \textsc{count\_walks} applies Equation~\eqref{eq:sigma_overline} to compute $\overline{\sigma}_{s(v,t)}$ for all $v \in V$ and $t \in [T]$ and then applies Equation~\eqref{eq:sigma} to compute $\sigma_{sv}$ for all $v \in V$ and $\delta_{sv}(v,t)$ for all $v \in V$ and $t \in [T]$. Finally, function \textsc{count\_walks} returns a dictionary $del$ containing the values of $\delta_{sv}(v,t)$ for all $v \in V$ and $t \in [T]$ and another dictionary $sig$ containing the values of $\sigma_{s(v,t)}$ for all $v \in V$ and $t \in [T]$.

For active walks, we finally need to ensure that the general contribution is also computed for temporal nodes not lying on the predecessor graph. We discuss this in the Appendix and show that the values of these temporal nodes can be computed on the fly during the general contribution recurrence provided that we order the predecessor graph which adds a factor in the final complexity of active walks.
\begin{algorithm}[H]
	\caption{\label{algo:betweenness} Computes the values of $B(v,t)$ for all temporal nodes}
	\begin{algorithmic}[1]
    \renewcommand{\algorithmicrequire}{\textbf{Input:}}
    \renewcommand{\algorithmicensure}{\textbf{Output:}}
    \Require  $G = (V,\E,T)$ : a temporal graph, $k$ : maximal waiting time for restless walks or $k = \infty$ for shortest walks.  
    \Ensure $B(v,t), \forall v \in V, t \in [T]$
    \Function{Betweenness\_centrality}{$G, k$}
    \State $B(v,t) = 0, \forall v\in V, t \in [T]$
    \For {$s \in V$}
    \State $pre,d$ = \textsc{Temporal\_BFS}$(G, s, k)$ \label{line:bfs} \Comment{Algorithm~\ref{algo:BFS}}
    \State $del, sig= \textsc{count\_walks}(pre,d,s)$ \Comment{Apply Corollary~\ref{cor:quan}} \label{line:count}
    \State $cum$ = \textsc{Contributions}$(G,pre,del, sig)$ \label{line:gen} \Comment{
    see Appendix~\ref{sec:appendix}}
    \State \textsc{Update\_betweenness}$(B,cum, del)$ \label{line:update} \Comment{Apply Equation~\eqref{eq:hat_b}}
     \EndFor
\State \textbf{return} $B$
    \EndFunction
  \end{algorithmic}
\end{algorithm}

\begin{proof}[Proof of Theorem~\ref{thm:gen}]
    For active and passive walks, the total cost of the predecessor graph construction in (Line~\ref{line:bfs}) of Algorithm~\ref{algo:betweenness} is for any node $s\in V$ in $O((n+m)T)$. Then we can compute all necessary quantities for the main recurrence (Line~\ref{line:count}). This is also done in $O((n+m)T)$ as well as explained in Corrolary~\ref{cor:quan}. The recurrence of Proposition~\ref{prop:rec_contri} computing all contributions from node $s \in V$ (Line~\ref{line:gen}) can be computed in $O((n+m)T)$ for passive walks. However, for active walks the same computations can be done and by ordering the predecessor graph we can compute contributions of temporal nodes not lying on the predecessor graph (see discussion in Appendix~\ref{sec:appendix}). The predecessor graph ordering costs then $O(nT^2)$ and the overall cost of (Line \ref{line:gen}) is $O(nT^2 + mT)$ for active walks. Finally, the application of the correction formula in Line~\ref{line:update} can be done in $O(nT)$.
\end{proof}


%% file: sections/mining.tex
\section{Experimental results}\label{sec:experiments}
For our experiments we built our algorithms on top of the code of \cite{buss2020algorithmic}. We implemented our algorithms with the different possible variants. We focus next on the variants of active and passive shortest walks, passive and active shortest $k$-restless walks for $k$ being equal to $10\%$ of the lifetime of the graph and passive shortest foremost walks.

We summarize here our main findings:
\begin{itemize}
    \item The active variants takes more time to be computed than the passive one as stated by Theorem~\ref{thm:gen} but this difference starts to emerge with large networks. Table~\ref{tab:gen_stat}.
    \item On the importance of times $B(t)$, the active variant points the importance of central times where central become more important in comparison with the passive variant. Figure~\ref{fig:time}.
    \item The rankings between (passive shortest), (active shortest) and the betweenness centrality computed on the aggregated static graph are positively correlated. There are more differences between these variants when looking at the intersection of the top ranked $10$-nodes. Figure~\ref{fig:heat}.
    \item Predicting the top $10$ ranked nodes by looking only at the first few interactions is much more accurate for passive shortest foremost variant compared to the others. Figure~\ref{fig:pred_inter_time}. 
\end{itemize}
\begin{table}[h!]
    \centering
      \resizebox{\textwidth}{!}{  
    \begin{tabular}{l|cccc|c|cc|cc|c|c}
    \hline  
dataset & nodes & events & edges & agg\_edges & Bu{\ss} & sh pas & sh act  & sh-rl act & sh-rl pas & sh-fm pas  & static \\ \hline
gre & $1547$ & $1300$ & $113112$ & $3680$ & $1924.1$  & $2921.2$  & $9067.7$  & $1496.3$  & $5535.6$  & $3839.8$ & $21.156$ \\
ren & $1407$ & $10825$ & $107384$ & $3718$ & $2480.7$  & $3343.4$  & $26144$  & $1656.3$  & $22379.$  & $3329.3$ & $29.046$ \\
bel & $1917$ & $1132$ & $120951$ & $6440$ & $2530.8$  & $3512.9$  & $12691.$  & $2056.5$  & $8111.3$  & $3753.4$ & $45.057$ \\
kuo & $549$ & $1211$ & $30545$ & $1952$ & $169.45$  & $242.67$  & $683.99$  & $119.67$  & $514.47$  & $240.69$ & $1.9935$ \\
prim & $242$ & $3100$ & $125773$ & $16634$ & $1146.2$  & $1701.6$  & $454.91$  & $761.39$  & $435.49$  & $1597.1$ & $4.0661$ \\
hs11 & $126$ & $5609$ & $28539$ & $3418$ & $75.741$  & $98.333$  & $101.61$  & $40.064$  & $89.428$  & $107.29$ & $0.3405$ \\
hs12 & $180$ & $11273$ & $45047$ & $4440$ & $208.56$  & $287.82$  & $371.63$  & $142.94$  & $384.21$  & $299.61$ & $0.9211$ \\
hp & $75$ & $9453$ & $32424$ & $2278$ & $121.75$  & $170.15$  & $76.980$  & $75.908$  & $59.653$  & $177.97$ & $0.1877$ \\
ht & $113$ & $5246$ & $20818$ & $4392$ & $43.157$  & $61.687$  & $63.644$  & $34.362$  & $65.114$  & $66.632$ & $0.2166$ \\
wp & $92$ & $7104$ & $9827$ & $1510$ & $12.020$  & $16.154$  & $51.370$  & $7.4984$  & $50.847$  & $17.592$ & $0.0782$ \\
    \end{tabular}
    }
    \caption{Statistics for the datasets. From left to right  number of nodes (nodes), number of times (events), number of temporal edges (edges), number of edges in the static representation of $G$. The execution time in seconds of (Bu{\ss}) implementation of \cite{buss2020algorithmic}, Algorithm~\ref{algo:betweenness} active shortest (sh act), passive shortest (sh pas),  active shortest $k$-restless (sh-rl act), passive shortest restless (sh-rl pas), passive shortest foremost (sh-fm) and on the aggregated static graph of $G$.}
    \label{tab:gen_stat}
\end{table}

Our code is open-source\footnote{\href{https://github.com/busyweaver/code_temporal_betweenness_}{github.com/busyweaver/code\_temporal\_betweenness\_}} and is written in \CC. 
We used an Intel(R) Xeon(R) Silver $4210R$ CPU $2.40$GHz without parallel processes. The datasets are divided into two types. Public transport datasets \cite{cities}. The datesets are:  gre (grenoble), ren (rennes), bel (belfast) and kuo (kuopio) and social contact traces from {\href{sociopatterns.org}{sociopatterns.org}} namely: prim (primaryschool), hs11 (HighSchool2011), hs12 (HighSchool2012), hp (Hospital Ward), ht (HyperText) and wp (Workplace). All datasets are available publicly.

On table~\ref{tab:gen_stat} we give information about the datasets that we used as well as running times of our algorithms and the one of \cite{buss2020algorithmic}. Our implementation is complementary to theirs since the authors of \cite{buss2020algorithmic} compute the overall betweenness centrality of nodes $B(v)$ for passive shortest walks. In comparison, our implementation computes all the values of $B(v,t)$, $B(v)$ and $B(t)$ for both active and passive variants. We note that according to Theorem~\ref{thm:gen} the active version is slower than its passive counterpart and the difference becomes more clear on larger graphs while on smaller networks the execution times are more comparable and for some instances the active version is slower than the passive one. This happens whenever the overall cost of ordering the predecessor graphs is less important than the overall gain in the sizes of the predecessor graphs which are smaller in the active case, see Remark~\ref{rem:small}.

\begin{figure}[h!]
    \hspace{-2.1cm}
    \includegraphics[scale=0.3]{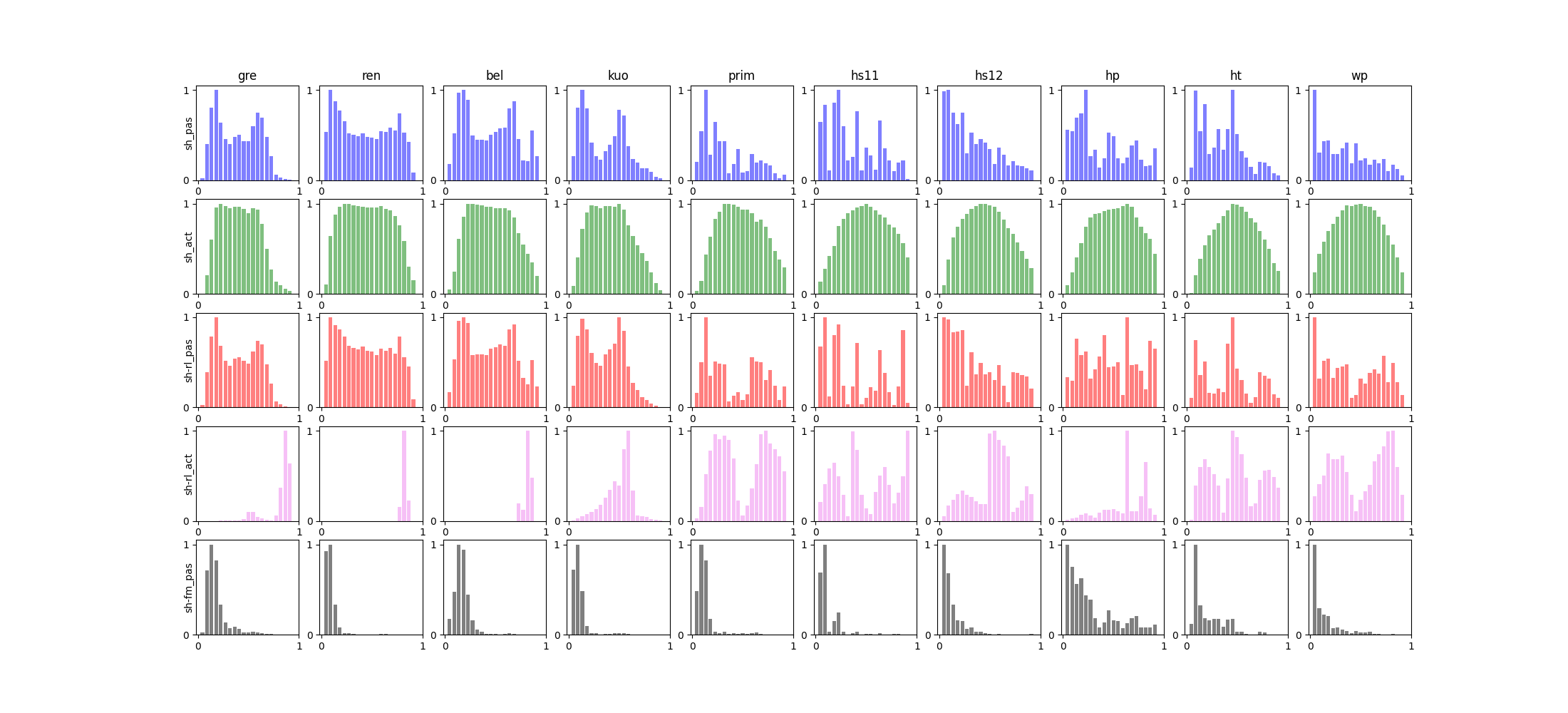}    
    \caption{Distribution of time centrality values $B(t)$ for the datasets. Each column represents a dataset. (1st row) correspond to the distribution of $B(t)$ for passive shortest, (2nd row) active shortest, (3rd row) passive shortest restless, (4th row) active shortest restless  and (5rd row) passive shortest foremost. The x-axis represents the renormalized life time of the temporal graph and the y-axis represents the values of $B(t)$ grouped into $20$ bars.}
    \label{fig:time}
\end{figure}

On Figure~\ref{fig:time}, we see that the distribution of the values of $B(t)$ are much more concentrated around central times (those in the middle of the temporal graph) for the shortest active walks than for shortest passive walks due to the contribution of intermediary node (first two rows) showing the difference between active and passive walks. For passive shortest walks important times tend to be in the beginning of the lifetime of the graph since many shortest walks can be formed when starting walks early in time and combining times later on. Finally, for passive foremost variant (last row) the most important times are seen in the beginning of the graph due to the fact that this measure is focused on temporal walks arriving the earliest.

\begin{figure}[h!]
    \centering
    \includegraphics[scale=0.7]{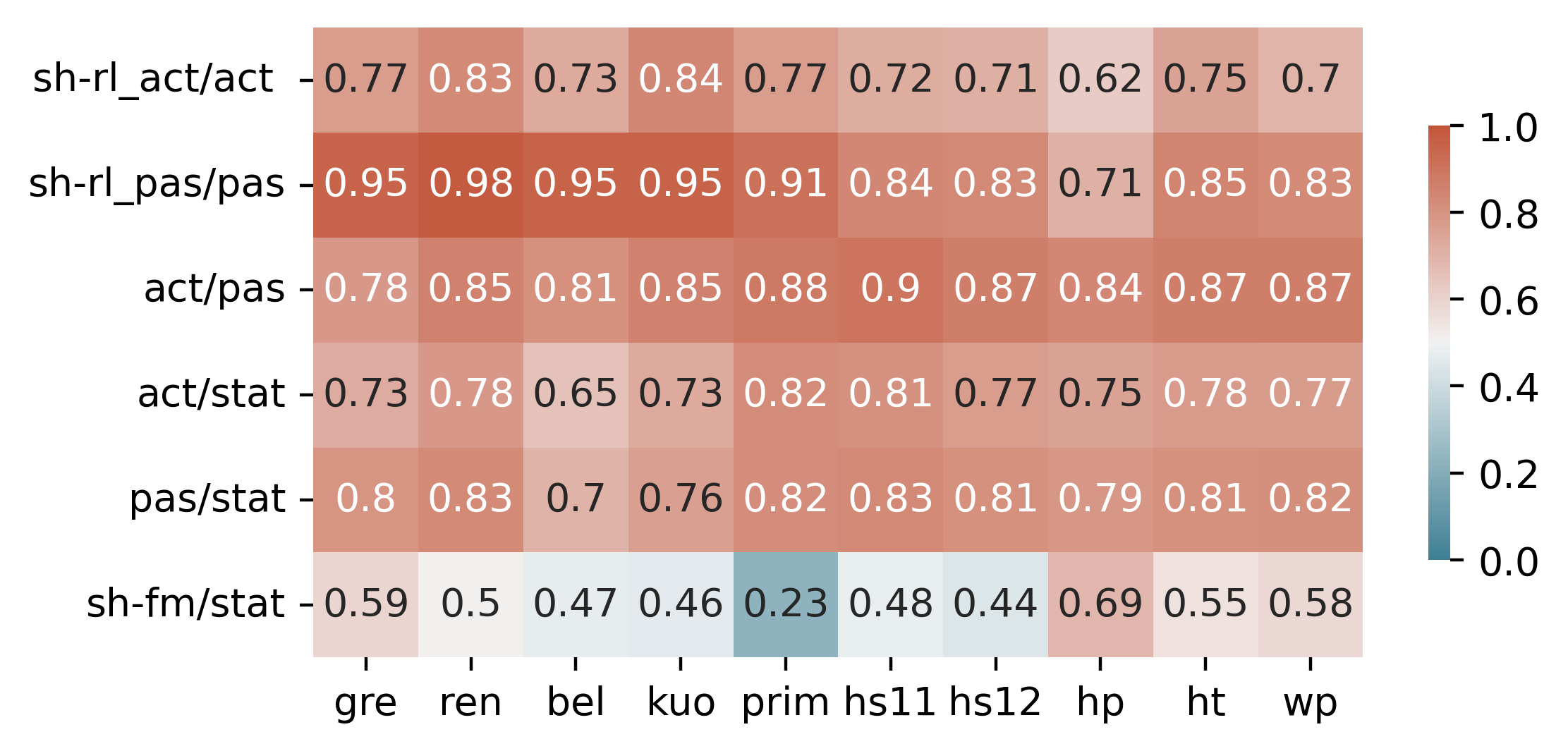}
    \includegraphics[scale=0.7]{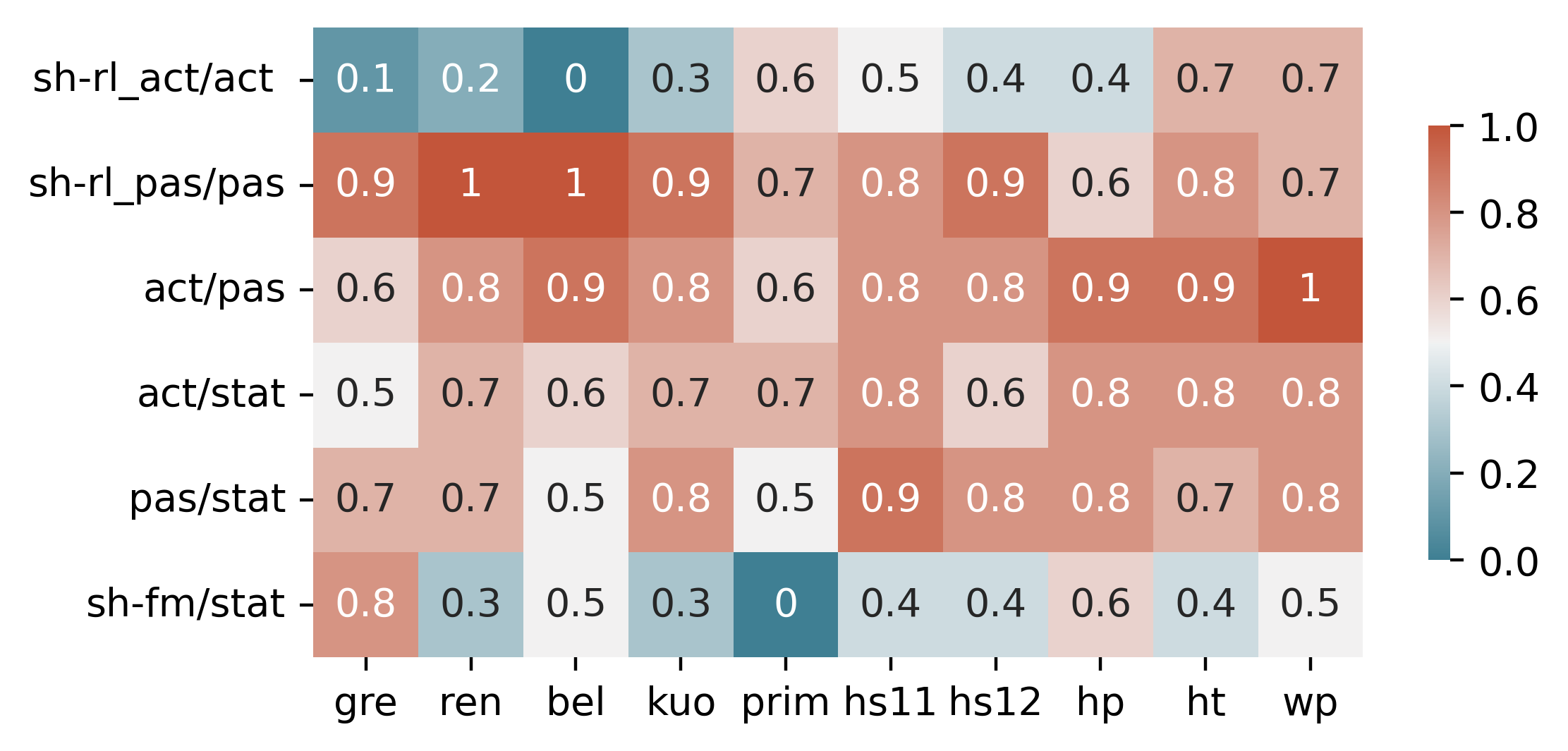}
    \caption{Heatmap of betweenness centrality of $B(v)$ comparisons of the datasets. (left) Kendall-tau rank correlation rankings and (right) intersection rate of the top $20$ nodes. In the figure act stands for active variant, pas for passive variant, sh-rl\_act active shortest restless, sh-rl\_pas passive shortest restless and stat for the static betweenness centrality on the aggregated graph.}
    \label{fig:heat}
\end{figure}
We compared the ranking correlation and intersection size of the different proposed variants together with the static betweenness centrality on the aggregated graph gives. We also compared shortest and shortest $k$-restless variants together.  On Figure~\ref{fig:heat} we see that the rankings and intersection of top $20$ nodes of $B(v)$ for several real world datasets show that passive betweenness centrality and the static one have a high correlation for all our datasets, always higher than the correlation between the active and static one (rows 4 and 5). Our results suggest that if we want an approximate ranking of $B(v)$ for passive (and to a lesser extent active) shortest walks, the static betweenness centrality gives a good approximation to it and runs much faster ($100$ times faster for all our datasets) than the temporal version. While the comparison between active temporal betweenness centrality and the static one show less correlation in general. The comparison between active and passive variants shows high correlation for $B(v)$ while the behaviour is largely different for $B(t)$ between active and passive variants. We also notice that the ranking correlation and intersection size is higher between passive variant of shortest and shortest $k$-restless than it is with its active counterpart (first two rows). If we care about arriving first to nodes (passive shortest foremost version) we see that the correlations with the betweenness on the static aggregated graph has a low correlation (last row).



Finally, in many practical applications we have access only to the first few interactions of the graph but we still want to predict the node rankings. Here we focus on predicting the overall node ranking $B(v)$ rather than the temporal one $B(v,t)$ - in fact it has been argued in a close context that predicting the temporal node centrality evolution is difficult \cite{magnien} -. In order to do so we introduce 
\begin{definition}[The graph of first $\mu$ times of $G$]\label{def:mug}
Let $G = (V, \E, T)$ be a temporal graph and let $\mu \in [0,1]$, and let  then $G^{\leq \mu} = (V, \E', \mu\,T)$ with $\E' = \{(u,v,t) \,|\, (u,v,t) \in \E, t \leq \mu\,T \}$.
\end{definition}
\begin{figure}[h!]
    \hspace{-2.1cm}
    \includegraphics[scale=0.3]{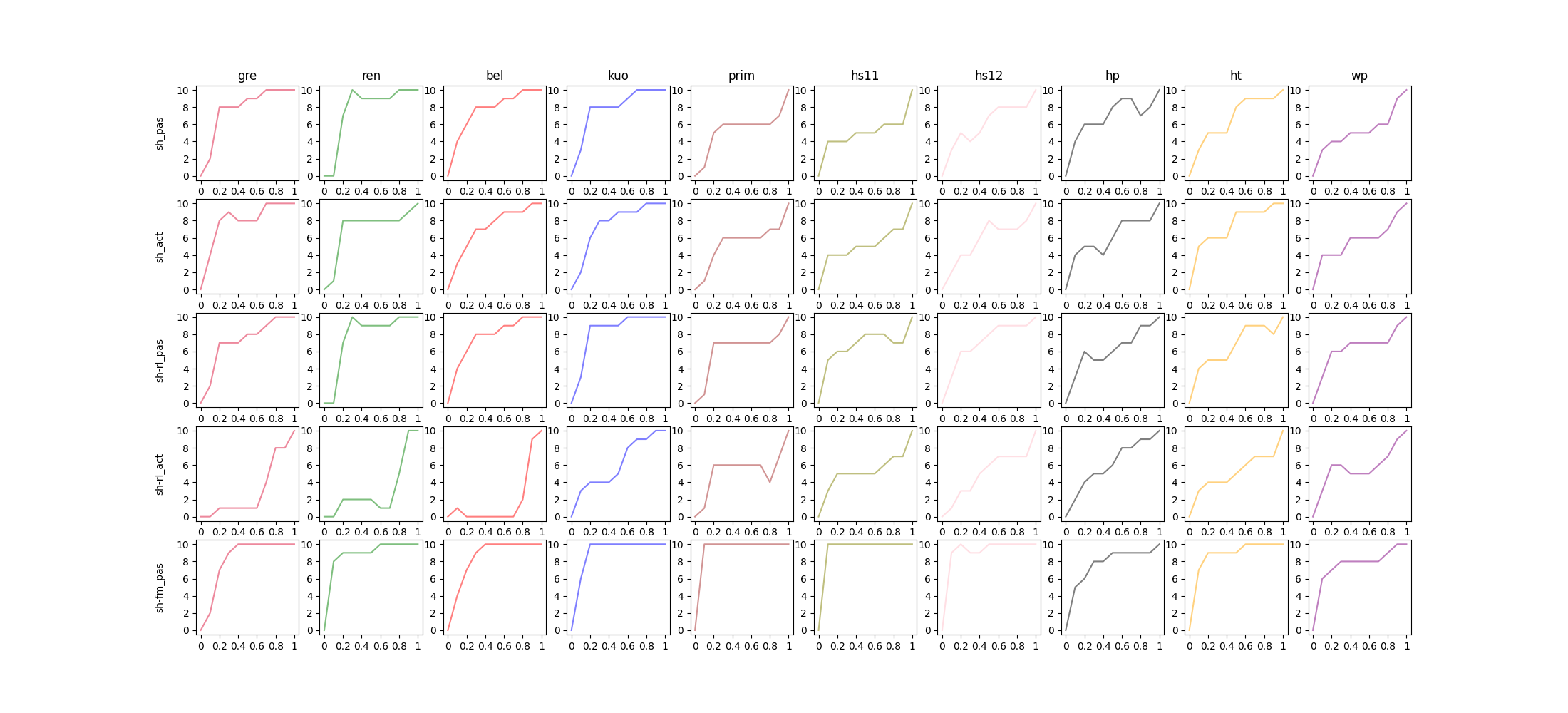}    
    \caption{Each column corresponds to a dataset. (1st row) correspond to the distribution of $B(t)$ for passive shortest, (2nd row) active shortest, (3rd row) passive shortest restless, (4th row) active shortest restless  and (5rd row) passive shortest foremost. Each graph has on its x-axis the $\mu$ values and on its y-axis the size of the intersection between top $10$ ranked nodes of nodes $B_G(v)$ of the temporal graph $G$ and top $10$ ranked nodes of the $B_{G^{\leq \mu}}(v)$.}
    \label{fig:pred_inter_time}
\end{figure}
Our results are summarized in Figure~\ref{fig:pred_inter_time} in which each plot corresponds to how well are the top $10$ nodes ranked when looking only at the first times of the graph. Therefore, the faster each plot reaches $10$ the fewer interactions we need to observe to correctly identify the top $10$ ranked nodes. It turns out that this depends on the criteria that we consider. For instance, considering passive shortest foremost variant (3rd row) looking at the first $10\%$ interactions gives a good approximation of the top ranked nodes. In fact this is in accordance with the last row of Figure~\ref{fig:time} where we see that most important times $B(t)$ in the graphs are the ones in the beginning of it.

%% file: sections/appendix.tex
\section{Details of the proofs}\label{sec:appendix}
A path in the predecessor graph will be denoted denoted $(s,t_1) \rightarrow (v_1,t_1) \rightarrow \dots \rightarrow (v_n,t_n)$. We also use the operator $\oplus$ for the concatenation of walks.  
\begin{proof}[Proof of Lemma~\ref{lem:bij}]
    By induction we show that a path in the predecessor graph $G_s$ starting at node $s$ corresponds to an exact optimal $s-(v,t)$ walk. Let $p$ be a path in the predecessor graph. Then $p \equiv p' \rightarrow (v,t)$ for some path $p'$. Let $(w,t')$ be the last node in $p'$. By induction hypothesis, suppose that $p'$ corresponds to an exact optimal $s-(w,t')$ walk $P'$. Since $((w,t'),(v,t)) \in E_s$, this implies that there exists an exact optimal $s-(v,t)$ walk $M$ that passes through $(w,t')$ before arriving to $(v,t)$. Let $M \equiv M' \oplus (w \overset{t}{\rightarrow}  v)$. $M'$ is necessarily an exact optimal $s-(w,t')$ walk since else $M$ would not be an exact optimal $s-(v,t)$ walk. Now $M'$ and $P'$ are both exact optimal $s-(w,t')$ walks implying that they have the same length. Since $M$ extends $M'$ with a single edge then $c_{s}(v,t) = c_{s}(w,t')+1$. Therefore, $P \equiv P' \oplus (w \overset{t}{\rightarrow}  v)$ is an exact optimal $s-(v,t)$ walk as well. 
    On the other hand, We show by induction that an exact optimal $s-(v,t)$ walk $P$ corresponds to a single path in $G_s$ starting at $(s,i)$ for some time $i$. Let $(w,t')$ be the last node appearance before $(v,t)$. Then $P \equiv P' \oplus (w \overset{t}{\rightarrow}  v)$. $P'$ is an exact optimal $s-(w,t')$ walk. Then by induction hypothesis, let $p'$ be the corresponding path in predecessor graph. Then $((w,t'),(v,t)) \in E_s$ and the path $p \equiv p' \rightarrow (v,t)$ is a path in the predecessor graph.
\end{proof}

\begin{proof}[Proof of Lemma~\ref{lem:dag}]
Using Lemma~\ref{lem:bij}, there are no exact optimal $s-(v,t)$ walks containing cycles because else we could immediately construct an exact $s-(v,t)$ walk $W$ without cycles and $len(W) < c_{s}(v,t)$ which is impossible. Hence the predecessor graph from $s$ is a acyclic.
\end{proof}

\begin{proof}[Proof of Proposition~\ref{prop:overline_sig}]
We know that $\overline{\sigma}_{s(v,t)}$ corresponds to the number of walks in the predecessor graph from $s$ ending in $(v,t)$ by Proposition~\ref{lem:bij}. By induction any optimal exact $s-(v,t)$ walk comes from a predecessor $(w,t')$ with  $(w,t') \in pre_s(v,t)$. Each optimal exact $s-(w,t')$ walk can be extended uniquely by appending $(w \overset{t'}{\rightarrow} v)$ to it and make it an optimal exact $s-(v,t)$ walk.
  \end{proof}

  \begin{proof}[Proof of Proposition~\ref{prop:sig}]
For passive walks, the set $\W_{s(v,t)}$ of optimal-$s-(v,t)$ walks and the set $\overline{\W}_{s(v,t)}$ of exact optimal-$s-(v,t)$ walks coincide. Therefore, $\sigma^{pas}_{s(v,t)} =\overline{\sigma}_{s(v,t)}$. For active walks, since an exact optimal $s-(v,t)$ walk $W$ can still be an optimal $s(v,t')$ walk for $t'>t$ if $c_{s}(v,t) = c_{s(v,t')}$ since the walks extend after their last transition.
\end{proof}

\begin{proof}[Proof of Proposition~\ref{prop:BFS}]
We show that at the start $i$-th iteration of Line~\ref{line:loop} in Algorithm~\ref{algo:BFS},  all temporal nodes which are exactly reachable have $\mathrm{dist}[v][t] = c_{s}(v,t) = i$, and for temporal nodes $(v,t)$ with an exact optimal $s-(v,t)$ walk $(v,t) \in Q$ and $pre[v][t] = pre_s(v,t)$ and $(v,t)$ is added exactly once to the queue. The property is true just before entering the loop the first time. Only temporal nodes $(s,t)$ where $(s,w,t) \in \E$ for some $w \in V$ are in the queue. The condition (not ($(s=a)$ and $(t' \neq t)$)) can only be met during the first iteration of the main loop. It ensure that all paths of length $1$ created have their first appearance time at the time of their first transition. Suppose that the property holds for iteration $i-1$. Then at iteration $i$. All temporal nodes in $Q$ are such that $\mathrm{dist}[v][t] = c_{s}(v,t) = i-1$ and each of them has an exact optimal $s-(v,t)$ walk to it. Each one of them is scanned for its outgoing neighbors $(w,t')$ and relaxed using function \textsc{relax} of Algorithm~\ref{algo:BFS}. 

If $(w,t')$ was never reached neither with an exact walk arriving to it or with an extended walk (only happens in active type). Then $dist[w][t'] = \infty$ and $(w,t')$ is added to the queue and $dist[w][t']$ set to $i$. If $(w,t')$ was reached before then $dist[w][t']<\infty$. If it was reached with a prior loop $j < i $ nothing happens. If it was reached exactly by a prior $(z,t'')$ in the same loop, then $(w,t')$ is not added to the queue and $dist[w][t']$ is set to $i$ since neither branches of condition in Line~\ref{line:cond} is met. Finally, if $(w,t')$ was reached by an extended $(w,t'')$ of the same loop $i$ then $dist[w][t']$ is set to $i$ and $pre[w][t'] = \{ \}$. Then the condition ($dist[b][t'] \geq \ell$ and $|\,pre[b][tp]\,| = 0$) is met and $(w,t')$ is added to the queue and its predecessor list will have size $>0$ ensuring $(w,t')$ is not added another time. In all cases $(w,t')$ is added exactly once. By definition all optimal exact $s-(w,t')$ walks have the same length and therefore all predecessor's of $(w,t')$ are added in the same loop. Finally all predecessors of $(w,t')$ arise in the same iteration $i$ and since $dist[w][t'] = i$ they are all added in the predecessor set. 
\end{proof}

\begin{proof}[Proof of Proposition~\ref{prop:com_pred}]
By using a queue each temporal node $(v,t) \in ER_s$ is scanned at most one time by \textsc{temporal\_bfs} of Algorithm~\ref{algo:BFS} as it was shown in Proposition~\ref{prop:BFS}. Then the same temporal arc $(v,w,t) \in \E$, can be relaxed up to $T$ times in Line~\ref{line:loop} of Algorithm~\ref{algo:BFS} and $T$ other times in Line~\ref{line:arc2} in function \textsc{relax} of Algorithm~\ref{algo:BFS}. Remember that each temporal nodes in $ER_s$ is added at most once to the queue. Thus the overall time complexity of Algorithm~\ref{algo:BFS}  is $O( m\, T +  n\,T)$.
\end{proof}

\begin{proof}[Proof of Corollary~\ref{cor:quan}]
For passive walks, $\overline{\sigma}_{s(v,t)}$ and hence ${\sigma}_{s(v,t)}$ can be computed recursively from the predecessor graph. Then $\sigma_{sz}$ and $\sigma_{sv}(v,t)$ can be computed using Equation~\eqref{eq:sigma_sz}. For active walks $\overline{\sigma}_{s(v,t)}$ can be computed recursively from the predecessor graph and then ${\sigma}_{s(v,t)}$ can be computed using Equation~\eqref{eq:sigma}. $\sigma_{sz}$ and $\sigma_{sv}(v,t)$ can be computed using Equation~\eqref{eq:sigma_sz} as in the passive case.
\end{proof}

\begin{definition}[arc dependency]
Fix a node $s$ and a type of walks. Then, $\delta_{sz}(v,t,(v,w,t')$ denotes the fraction of optimal $s-z$-walk
in $\W$ that go through the node appearance $(v,t)$ and then use the temporal arc $(v, w, t') \in \E$.
\end{definition}

\begin{lemma}\label{lem:suff} Let $G=(V,\E,T)$ be a temporal graph, fix a type of walks and a node $s \in V$. Let $G_s=(V_s,E_s)$ be the predecessor graph from $s$. Let $(v,t)$ be a temporal node and $(v,w,t') \in \mathcal{E}$. If $\delta_{sz}(v, t, (v, w, t' )) > 0$ , then
  \[\delta_{sz}(v,t,(v,w,t')) = 
    \dfrac{\sigma_{s(v,t)}}{\sigma_{s(w,t')}}
    \dfrac{\sigma_{sz}(w,t')}{\sigma_{sz}}.\]
\end{lemma}
\begin{proof}[Proof of Lemma~\ref{lem:suff}]
  For $\mathrm{passive}$ walks, only temporal nodes $(v,t) \in V_s$ can have strictly positive values of $\delta_{sz}(v, t, (v, w, t' ))$. The proof then corresponds to the one in
  \cite{rymar2021towards} by noticing that the fraction $\frac{\sigma_{sz}(w,t')}{\sigma_{s(w,t')}}$ corresponds to the number of optimal suffixes starting at $(w,t')$ and ending in $z$. For $\mathrm{active}$ walks the fraction $\frac{\sigma_{sz}(w,t')}{\sigma_{s(w,t')}}$ also corresponds to the number of optimal suffixes starting at $(w,t')$ and ending in $z$. However, if in Definition~\ref{def:opt_walk}, we would not have extended the optimal walks to $W_T$, this fraction would not correspond to the suffixes when $z = w$. Apart from this detail the rest of the proof also follows from the same reference.
\end{proof}

\begin{proof}[Proof of Proposition~\ref{prop:rec_contri}]
  The proofs closely follows the one in \cite{rymar2021towards} by using Lemma~\ref{lem:suff}. In the passive case the proof is the same. In the active case, we notice that $(v,t)$ might not belong to $V_s$. Therefore all shortest paths passing through $(v,t)$ will pass through the first $t''<t$ with $(v,t'') \in V_s$. Hence, the index of the sum looks at the successors of $(v,t'')$ in $G_s$ and only need to consider those successors $(w,t')$ with $t' \geq t$ so that these paths pass through $(v,t'')$ then $(v,t)$ and then go to $(w,t')$. The rest of the proof follows \cite{rymar2021towards}.
\end{proof}

\subsection{Discussion of general contribution for temporal nodes not lying on the predecessor graph}\label{sec:discussion}

For active walks, let $G_s = (V_s,E_s)$ be the predecessor graph from $s$. Let $t' = before_{G_s}(v,t)$. Then $\sigma_{s(v,t)} = \sigma_{s(v,t')}$. This result can be seen since we know that $t' \leq t$, and $(v,t') \in V_s$. If $t'= t$ the result is immediate. If $t'<t$, then $(v,t) \notin V_s$, and therefore there are no exact optimal $s-(v,t)$ walks. All the optimal $s-(v,t)$ walks arrive from $t'$. As a consequence:
\begin{equation}\label{eq:rec_act_inter}
      \delta^{act}_{s \bullet}(v,t) = \delta^{act}_{sv}(v,t) + \sum\limits_{\substack{ t'':= before_{G_s}(v,t)\\ (w,t') \in succ_{s}
      (v,t'') \\ t' \geq t}}\dfrac{\sigma^{act}_{s,(v,t'')}}{\sigma^{act}_{s,(w,t')}} \delta^{act}_{s
    \bullet}(w,t'),
  \end{equation}
This last Equation ensures that for active walks the computation $\delta^{act}_{s\bullet}(v,t)$ for temporal nodes $(v,t) \notin V_s$ can be done on the fly while computing $\delta^{act}_{s\bullet}(v,t')$ with $t' = before_{G_s}(v,t)$. This implies computing the elements $(w,t'')$ of the successor set of $(v,t)$ in a decreasing order of time and give the value $\delta^{act}_{s\bullet}(v,t)$ before the sum has completed since we need to stop when the elements $(w,t'')$ have $t''<t$.

\begin{algorithm}[H]
     \begin{algorithmic}[1]
 \renewcommand{\algorithmicrequire}{\textbf{Input:}}
    \renewcommand{\algorithmicensure}{\textbf{Output:}}
    \Require  $G = (V,\E,T)$ : a temporal graph, $G_s = (V_s,E_s)$ the predecessor graph of $s$, $del$ the values of $\delta{sv}(v,t)$ for all $(v,t)$ and $sig$ the values of $\sigma_{s(v,t)}$ for all $(v,t)$ 
    \Ensure A dictionary $cum$ containing the values of  $cum[(v,t)] = \delta_{s\bullet}(v,t), \forall v \in V, t \in [T]$
    \Function{Contributions}{$G,G_s,del, sig$}
    \State $cum(v,t) = 0, \forall v \in V, t \in [T]$
    \State $visited = \{\}$
    \For{$(v,t) \in sources(G_s)$} \Comment{sources are nodes with no incoming edges}
    \State $\textsc{general\_rec}((v,t),cum,G,G_s,del, sig,visited )$
    \EndFor
 \State \textbf{return} $cum$
    \EndFunction
  \end{algorithmic}
 	\begin{algorithmic}[1]
    \Function{General\_rec}{$(v,t),cum,G,G_s,del, sig,visited$}
    \If{$(v,t)$ not in visited}
    \State $su = 0$
    \For{$  t' \in \{t'' | \exists ((v,t), (w,t'')) \in E_s \}  $, in decreasing order} \Comment{Ordering is necessary only for active walks}
    \For{$ w \in \{((v,t),(w,t')) \in E_s  \}$}
    \State \textsc{General\_rec}{$((w,t'),cum,G,G_s,del, sig,visited)$}
    \State $su = su + \dfrac{sig[(v,t)] }{sig[(w,t')]} cum(w,t')$
    \EndFor
    \State $\textsc{inter\_contribution}((v,t), t', G_s, cum, del ,su,visited)$ \Comment{for passive walks ignore this instruction}
    \EndFor
    \State $cum[(v,t)] = su +del[(v,t)]$
    \State $visited.\textsc{add}((v,t))$
    \EndIf
    \EndFunction
        \Function{inter\_contribution}{$(v,t),t', G_s, cum,del, su,visited$} 
    \State $t'' = t' - 1$
        \While{$t'' \geq t$ and $before_{G_s}(v,t'') = t$ and $(v,t'') \notin visited$}
        \State $cum[(v,t'')] = del[(v,t'')] + su$
        \State $visited.\textsc{add}((v,t''))$
    \State $t'' = t''-1$ 
    \EndWhile
    \EndFunction
  \end{algorithmic}
  	\caption{\label{algo:general_contribution} Compute the values of $\delta_{s\bullet}(v,t)$ for a temporal graph $G$}
\end{algorithm}

\subsection{Passive shortest foremost and strict variants}\label{sec:shfm}
For passive shortest foremost variant. We can define:
\[ c^{fm,pas}_s(v,t) = \min\limits_{ W \in W_{s(v,t)}^{pas,\infty} } (arr(W)  n +   len(W)) \]
If the set of $s-v$ walks is empty then $ c^{fm,pas}_s(v,t) = \infty$. In this way the values of $c^{fm,pas}_s(v)$ as defined in Equation~\eqref{eq:cost3} coincide with $\min_{t \in [T]} c^{fm,pas}_s(v,t)$. We notice that the predecessor graph $G$ from node $s$ in the passive case of sh-$\infty$ is the same as the predecessor graph $G'$ of sh-fm in the passive case. This comes from the fact that all $s-(v,t)$ walks have the same arrival time $t$ (remember that we consider only passive type in the foremost setting). Therefore, $c^{fm,pas}_s(v,t) = c^{\infty,pas}_s(v,t) +t\cdot n$ for all temporal nodes $(v,t)$. However, $\sigma_{sv}$ and $\sigma_{sv}(v,t)$ will be different since $c^{fm}_s(v) \neq c^{sh-\infty}_s(v)$ in general. The same results and proofs then hold in the same way except the values of $B(v,t)$ become different. For strict version of all variants considered. The only change to be made in Algorithm~\ref{algo:BFS} is on Line~\ref{line:strict} by replacing $t'\geq t$ with $t'>t$. The extension of the recurrence in Proposition~\ref{prop:rec_contri} is immediate as it is the case in \cite{rymar2021towards}.

%% file: sections/conclusion.tex
\section{Conclusion}\label{sec:conclu}
The main reason why our formalism does not hold on the active variant of shortest foremost walks is the lack of prefix optimality in this variant. For instance, on the graph of Figure~\ref{fig:ex}, $c_a(d,6) = 30 + 3 = 33$ and the walk $W = a \overset{1}{\rightarrow}  b
\overset{5}{\rightarrow} c \overset{6}{\rightarrow} d$ and $W \in \W_{a(d,6)}$, since $arr(W)\cdot n +len(W)= 33$. on the other hand $c_a(c,5) = 10 + 2 = 12$ while $W' = a \overset{1}{\rightarrow}  b \overset{5}{\rightarrow} c$ and $W' \notin \W_{a(c,5)}$ since $arr(W')\cdot n + len(W') = 27$ and therefore the predecessor graph does not account exactly for the set $\overline{\W}_{s(v,t)}$ (see Lemma~\ref{lem:bij}) as it is the case for the other two variants whether on active or passive walks.
Our results leave an open question on whether it is possible to characterize cost functions that can be solved using a temporal BFS as we did for the three variants of this paper in the same vein of \cite{rymar2021towards}.

Our results improve the theoretical time analysis of previously known methods to compute the temporal betweenness centrality on shortest paths variants. It would be interesting to know if these results could be improved or if it is not the case to extend known hardness complexity results on static betweenness centrality such that \cite{hab} to the temporal case.

Another direction is to look for guaranteed approximations to the temporal betweenness centrality which started to be studied recently in \cite{onbra,cruciani}. 

Finally, the temporal betweenness centrality has been defined in different time dependent formalisms such that Stream Graphs \cite{latapy2018stream,bet_link} that allow for continuous time and it would be interesting to find out if the same kind of results hold in that setting as well.
